\documentclass[letterpaper, 10 pt, conference]{ieeeconf}  

\IEEEoverridecommandlockouts                              

\usepackage{graphicx}
\usepackage{caption}
\graphicspath{
    {figures/}
}

\usepackage{amsmath,amssymb,enumerate}

\usepackage{amsthm}
\usepackage{setspace}
\usepackage{booktabs}
\usepackage[usenames,dvipsnames,svgnames,table]{xcolor}
\usepackage{mathtools}
\usepackage{algorithm, algorithmicx, algpseudocode}
\usepackage{blindtext}
\usepackage{gensymb}
\usepackage{xparse}
\usepackage{lipsum}
\usepackage{mathrsfs}
\usepackage[mathscr]{euscript}
\usepackage{times}
\usepackage{cite} 
\usepackage{multicol}
\definecolor{darkgreen}{rgb}{0,0.6,0}
\usepackage[bookmarks=true,colorlinks=true,pdfpagemode=UseNone,citecolor=darkgreen,linkcolor=black,urlcolor=BrickRed,pagebackref]{hyperref}
\usepackage[caption=false,font=footnotesize]{subfig}
\usepackage{amsfonts}
\usepackage[utf8]{inputenc}
\usepackage[T1]{fontenc}
\usepackage{textcomp}
\usepackage{arydshln}
\usepackage{balance}

\newtheorem{problem}{Problem}
\newtheorem{theorem}{Theorem}

\renewcommand{\thefigure}{\arabic{figure}}

\definecolor{note}{rgb}{0.1,0.1,1}
\definecolor{rephase}{rgb}{0.15,0.7,0.15}
\definecolor{bag}{rgb}{0.6,0.6,0.2}

\makeatletter
\renewcommand*\env@matrix[1][c]{\hskip -\arraycolsep
  \let\@ifnextchar\new@ifnextchar
  \array{*\c@MaxMatrixCols #1}}
\makeatother

\newcommand{\transpose}{\mathsf{T}}

\makeatletter
\newcommand{\mathleft}{\@fleqntrue\@mathmargin0pt}
\newcommand{\mathcenter}{\@fleqnfalse}
\makeatother

\usepackage{hyperref}
\hypersetup{
	colorlinks=true, 
	linktoc=all,     
	linkcolor=black,  
}

\definecolor{orange}{RGB}{255,127,0}

\title{\LARGE \bf A Generalized Metriplectic System via Free Energy and System~Identification via Bilevel Convex Optimization}

\author{Sangli Teng, Kaito Iwasaki, William Clark, Xihang Yu, Anthony Bloch, Ram Vasudevan, Maani Ghaffari
\thanks{S. Teng, K. Iwasaki, W. A. Bloch, R. Vasudevan and M. Ghaffari are with the University of Michigan, Ann Arbor, MI 48109, USA. \texttt{\{sanglit, kaitoi, abloch, ramv, maanigj\}@umich.edu}. W. Clark is with Ohio University, Athens, OH 45701. \texttt{clarkw3@ohio.edu}. X. Yu is with MIT. \texttt{jimmyyu@mit.edu} }
}

\begin{document}

\maketitle
\thispagestyle{empty}
\pagestyle{empty}

\begin{abstract}
This work generalizes the classical metriplectic formalism to model Hamiltonian systems with nonconservative dissipation. Classical metriplectic representations allow for the description of energy conservation and production of entropy via a suitable selection of an entropy function and a bilinear symmetric metric. By relaxing the Casimir invariance requirement of the entropy function, this paper shows that the generalized formalism induces the \emph{free energy} analogous to thermodynamics. The monotonic change of \emph{free energy} can serve as a more precise criterion than mechanical energy or entropy alone. This paper provides examples of the generalized metriplectic system in a 2-dimensional Hamiltonian system and $\mathrm{SO}(3)$. This paper also provides a bilevel convex optimization approach for the identification of the metriplectic system given measurements of the system. 

\end{abstract} 

\IEEEpeerreviewmaketitle

\section{Introduction}
A metriplectic system combines the structure of the Poisson bracket and a bilinear symmetric bracket in 
a suitable fashion. The classical formalism of the metriplectic system induces the first and second laws of thermodynamics, i.e., the conservation of energy and production of entropy. To achieve energy conservation, entropy is chosen as a Casimir function, whose flow vanishes with respect to the Poisson bracket. The bilinear symmetric bracket is constructed to ensure the gradient flow of Hamiltonian vanishes while the entropy increases monotonically. 

However, this special entropy construction and the bilinear bracket cannot model general complex dissipation, such as rigid body motion in fluids \cite{o2022neural, song2023uncertainty, song2024turtlmap}. This work introduces a generalized metriplectic system by relaxation of the classical formalism in both entropy construction and bracket. We show that the generalized system induces the \emph{free energy}, whose counterpart in thermodynamics is a more precise criterion than the internal energy or the entropy alone. 

In particular, the main contribution of this work can be summarized as follows.
\begin{enumerate}
    \item Introduction of a generalized metriplectic system and induced free energy via the relaxation of the classical formalism. 
    \item A system identification approach that identifies the metric and entropy simultaneously via bilevel convex programming.
    \item Contructive examples of the proposed metriplectic system for a 2D system and a system evolving on $\mathrm{SO}(3)$.
    \item The software for reproducing the presented results can be downloaded at \href{https://github.com/UMich-CURLY/A_Generalized_Metriplectic_System.git}{link}. 
\end{enumerate}


 \begin{figure}[t]
    \centering
    \includegraphics[width=1\linewidth]{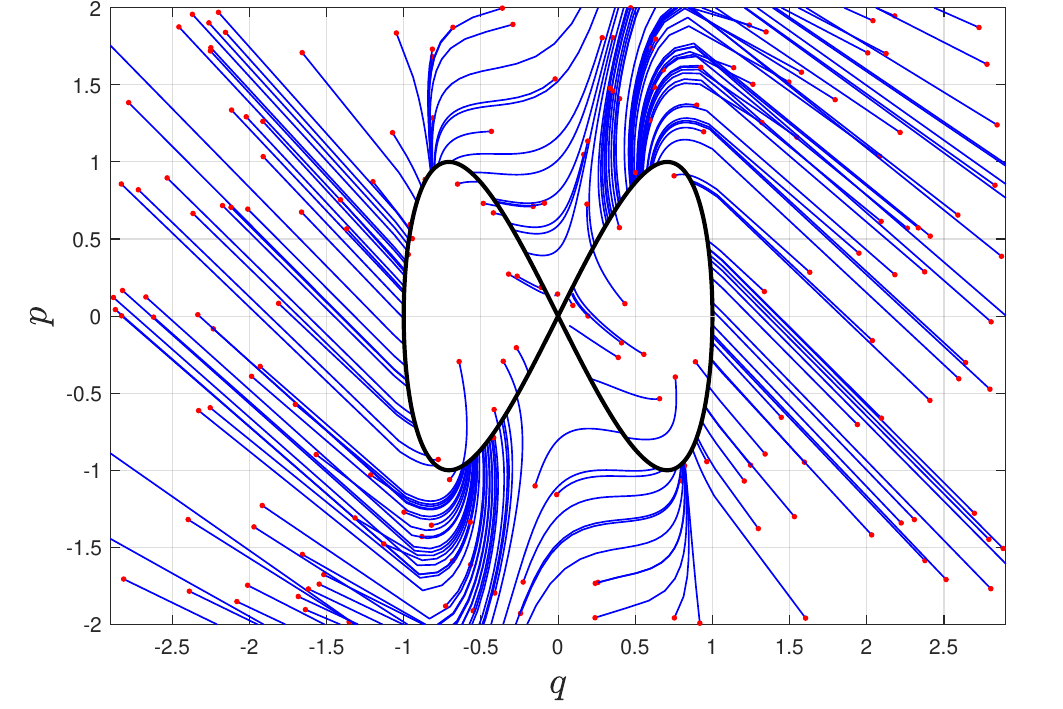}
    \caption{Trajectory of the 2D metriplectic system in the phase space. The trajectories converge to the zero-level set of the free energy ($\infty$ shape) via a selected entropy function and metric. The red dots denote the randomly sampled initial conditions.}
    \label{fig:linear-traj}
\end{figure}

\section{Related work}
The metriplectic formalism provides a natural framework for combining Hamiltonian dynamics, represented by the Poisson bracket, with dissipative dynamics, represented by a symmetric bilinear bracket. One motivation in the development of metriplectic formalism is the unification of dissipation into well-developed Hamiltonian systems, allowing for modeling systems that exhibit both energy conservation and irreversible processes. Early work in geometric Hamiltonian mechanics, such as Souriau \cite{souriau1970dynamiques} and Abraham and Marsden \cite{AbrahamMarsden}, provide the foundation for much of the later developments in this field. On the other hand, the study of dissipation in Hamiltonian systems has received less attention, largely due to the broader and more varied nature of dissipative phenomena. 

Several efforts have been made to formulate several classes of dissipative systems in the context of gradient flows. The double-bracket formalism, originally introduced by Brockett \cite{BROCKETT199179}, see also\cite{bloch1990steepest} and \cite{bloch1992completely}.
This was further extended to the mechanical setting by Bloch et al. \cite{BLOCH199437,Bloch1996}, which investigates a special type of gradient flow on adjoint orbits of Lie groups, where the dissipation is driven by a specific metric structure. In an infinite dimensional setting, diffusion-like PDEs are extensively studied to model the dissipative nature of complex physical systems \cite{Otto2001, Ambrosio2008,Villani2009,Jungel2016}. 

The metriplectic formalism, originally introduced by Morrison \cite{MORRISON1986410}, is another class of such dissipative systems that incorporates gradient flows along Hamiltonian dynamics. The term ``metriplectic" reflects the combination of a symplectic structure, which governs the conservative Hamiltonian part, with a metric structure that introduces dissipation. In classical metriplectic systems, Casimir invariants represent entropy, ensuring the conservation of the Hamiltonian while allowing for entropy production through dissipation. This approach enables the modeling of systems that simultaneously respect the first and second laws of thermodynamics, capturing both energy conservation and irreversible processes.

Applications of metriplectic systems have been explored in several contexts. In mechanical systems, metriplectic dissipation can be used for motion control of rigid bodies \cite{materassi2018metriplectic}. For Lie groups, metriplectic dissipation is considered a forcing term in the Lie-Poisson equation \cite{bloch2024metriplecticeulerpoincare}. From a control-theoretic perspective, the stabilization of metriplectic systems was studied in \cite{Birtea_2007} using an invariance principle. Also, as a nonlinear affine control problem, feedback stabilization of metriplectic systems was studied in \cite{HUDON201312}, invoking the Hodge decomposition of Brayton–Moser systems \cite{GUAY201266}.

Initially developed in the context of plasma physics, the formalism has since been applied to a wide range of physical systems. In nonequilibrium thermodynamics, the GENERIC formulation for thermodynamics, introduced by Grmela and Öttinger \cite{GENERIC1997GrmelaOttinger}, builds upon the foundational concepts of metriplectic systems initially developed by Kaufman \cite{KAUFMAN1984419} and Morrison \cite{MORRISON1986410}. This generalization led to applications in a wide range of disciplines. In plasma physics, the collisional effect in the Vlasov-Poisson system was studied in \cite{morrison1984brackets}, and dissipative visco-resistive magneto-hydrodynamics was studied in \cite{Materassi_2012}. In fluid dynamics, the compressible Navier-Stokes equations have been identified as metriplectic systems \cite{morrison1984brackets}, and the formalism has been extended to encompass general complex fluids \cite{GENERIC1997GrmelaOttinger}. More recent works include Korteweg-type fluids \cite{Suzuki_2020} and the Smoluchowski equation \cite{WAGNER2001177}. Applications have also been seen beyond plasma and fluids. For example, in statistical mechanics, the GENERIC formalism of Vlasov–Fokker–Planck
equation was investigated in \cite{Duong_2013}. 

In parallel to these developments, a more general formalism of metriplectic systems has been developed using triple-bracket structures \cite{bloch2013gradient}, encompassing both finite and infinite-dimensional systems. It considers a wide class of gradient flows arising from the normal metric on adjoint orbits of a Lie group and the K\"{a}hler metric. It uses triple-bracket structures to successfully generate wide classes of dissipative systems, including metriplectic $\mathfrak{so}(3)$ system and KdV equation with dissipation. More recently, a curvaturelike framework, called metriplectic 4-bracket structure, formulating Riemannian-like metriplectic geometry was studied in \cite{curvature4bracket2024Morrison}.

While the theoretical framework of metriplectic systems has been well-explored, computational developments in this area remain scarce, though there are a few developments worth mentioning. Based on discrete gradient methods, structure-preserving numerical integrators were developed in \cite{Ottinger2018}, \cite{shang2020structure}, and \cite{energymomentymnumerics2020}. A particular application of these methods to Riemannian-structure induced dissipation was discussed in \cite{bloch2024metriplecticeulerpoincare}. In a different line of research, metriplectic integrators through neural networks are studied in \cite{NEURIPS2021_2d1bcedd} and \cite{Zhang_2022}. Lastly, the computationally efficient reduced-order model, based on proper orthogonal decomposition, adapted to metriplectic systems, was developed in \cite{GRUBER2023115709}. 



\section{Preliminaries}
\label{prelim}

Metriplectic systems generalize the dynamics generated by the Poisson bracket with an additional symmetry bracket that represents dissipative effects. The metriplectic formalism provides a useful formulation to link a dynamical system to its environment. 

\subsection{Poisson structure}

Consider a differentiable manifold $P$ equipped with a Poisson structure given by a Poisson (bilinear) bracket:
\begin{equation}
    \begin{array}{clc}
C^{\infty}(P) \times C^{\infty}(P) & \longrightarrow & C^{\infty}(P), \\
(f, g) & \longmapsto & \{f, g\},
\end{array}
\end{equation}
satisfying the following properties:
\begin{enumerate}
    \item Skew-symmetry, $\{g, f\}=-\{f, g\}$;
    \item Leibniz rule, $\{f g, h\}=f\{g, h\}+g\{f, h\}$;
    \item  Jacobi identity, $$\{\{f, g\}, h\}+\{\{h, f\}, g\}+\{\{g, h\}, f\}=0;$$for all $f, g, h \in C^{\infty}(P)$.
\end{enumerate}

Given a Poisson manifold with bracket $\{\cdot,\cdot\}$ and a function $f \in C^{\infty}(P)$ we may associate with $f$ a unique vector field $X_f \in \mathfrak{X}(P)$, the Hamiltonian vector field defined by $X_f(g)=\{g, f\}$.
Moreover, on a Poisson manifold, there exists a unique bivector field $\Pi$, a Poisson bivector (that is, a twice contravariant skew symmetric differentiable tensor) such that
\begin{equation}
    \{f, g\}:=\Pi(d f, d g), \quad f, g \in C^{\infty}(P).
\end{equation}

\subsection{Positive Semi-Definite Inner Product}
Assume that for each point $x \in P$ we have a positive semidefinite inner product for covectors
\begin{equation}
    \mathcal{K}_x: T_x^* P \times T_x^* P \rightarrow \mathbb{R},    
\end{equation}
from which we can define $\sharp_{\mathcal{K}}: T^* P \rightarrow T P$ by
\begin{equation}
    \sharp \mathcal{K}\left(\alpha_x\right)=\mathcal{K}_x\left(\alpha_x, \cdot\right),    
\end{equation}
and the gradient vector field
\begin{equation}
    \operatorname{grad}^{\mathcal{K}} S=\sharp \mathcal{K}(d S),
\end{equation}
for any function $S: P \rightarrow \mathbb{R}$.
$\mathcal{K}$ defines a symmetric bracket given by
\begin{equation}
    (d f, d g)=\mathcal{K}(d f, d g).
\end{equation}




\subsection{Metriplectic System}

A metriplectic system consists of a smooth manifold $P$, two smooth vector bundle maps $\sharp_{\Pi}, \sharp \mathcal{K}: T^* P \rightarrow T P$ covering the identity, and two functions $H, S \in C^{\infty}(P)$ called the Hamiltonian (or total energy) and the entropy of the system, such that for all $f, g \in$ $C^{\infty}(P):$
\begin{enumerate}
    \item $\{f, g\}=\left\langle d f, \sharp_{\Pi}(d g)\right\rangle$ is a Poisson bracket ( $\Pi$ denotes the Poisson bivector).
    \item $(f, g)=\left\langle d f, \sharp_{\mathcal{K}}(d g)\right\rangle$ is a positive semidefinite symmetric bracket, i.e., $(\cdot, \cdot)$ is bilinear and symmetric.
    \item $\sharp_{\mathcal{K}}(d H)=0$ or equivalently $(H, f)=0, \forall f \in C^{\infty}(P)$.
    \item $\sharp_{\Pi}(d S)=0$ or equivalently $\{S, f\}=0, \forall f \in C^{\infty}(P)$, that is, $S$ is a Casimir function for the Poisson bracket.
\end{enumerate}

Consider the function $E=H+S: P \rightarrow \mathbb{R}$. Then the dynamics of the metriplectic system is determined by

\begin{equation}
    \begin{aligned}
\frac{d x}{d t} & =\sharp_{\Pi}(d E(x(t)))+\sharp_{\mathcal{K}}(d E(x(t))) \\
& =\sharp_{\Pi}(d H(x(t)))+\sharp_{\mathcal{K}}(d S(x(t))) \\
& =X_H(x(t))+\operatorname{grad}^{\mathcal{K}} S(x(t)),
\end{aligned}
\end{equation}

where $X_H=\sharp_{\Pi}(d H)$ and $\operatorname{grad}^{\mathcal{K}} S=\sharp \mathcal{K}(d S)$. From the equations of motion it is simple to deduce the following:
\begin{enumerate}
    \item First law: conservation of energy, $\frac{d H}{d t}=\{H, H\}+(H, S)=0$
    \item Second law: Entropy production, $\frac{d S}{d t}=(S, S) \geq 0$.
\end{enumerate}

Thus, metriplectic dynamics embodies both the first and second laws of thermodynamics. Denote parameterize the metric $\mathcal{K}$ as matrix $K$, we have the system in matrix form:
$$
\dot{x}=\Pi \nabla H+K \nabla S. 
$$



\section{Generalized Metriplectic System}
In this section, we relax the conditions in the previous section to model dissipation that unnecessarily conserve the energy, i.e, the case that the entropy is not a Casimir. 

The condition $\left\{S, f\right\} = 0, \forall f \in C^{\infty}(P)$ constraints that entropy is a Casimir function. The constraint on the metric $\left( H, f\right)= 0, \forall f \in C^{\infty}(P)$ is also chosen to ensure the energy is not dissipative.  However, modeling the entropy as Casimir makes it hard to account for complex dissipation that do not conserve energy. 

As a generalization, we drop these conditions and we have the generalized metriplecic system:
\begin{equation}
    \begin{aligned}
    \frac{d x}{d t} & =\sharp_{\Pi}(d E(x(t)))+\sharp_{\mathcal{K}}(d E(x(t))) \\
    \end{aligned}
\end{equation}
Letting $x = H$, we have:
\begin{equation}
\begin{aligned}
    \dot{H} &= \left\{dH, dH + dS\right\} +( dH, dH + dS) \\
    &= \left\{dH, dS\right\}+( dH, dH + dS),
\end{aligned}
\end{equation}
which does not necessarily equal $0$, thus making it possible to dissipate energy. Then we let $x = S$, which yields:
\begin{equation}
\begin{aligned}
    \dot{S} &= \left\{dS, dH + dS\right\} + ( dS, dH + dS) \\
    &= \left\{dS, dH\right\}+ ( dS, dH + dS). 
\end{aligned}
\end{equation}
To sum them up, we have:
\begin{equation}
\begin{aligned}
    \dot{E} &= \dot{H} + \dot{S} \\
    &=\left\{dH + dS, dH + dS\right\} + (dH + dS, dH + dS) \\
    &= (dE, dE) \ge 0. 
\end{aligned}
\end{equation}
Thus, we recover the free energy criterion for a thermal dynamics process. In this modeling, the dynamics will be such as to ensure the monotone change of free energy, which can be an analogy of the Gibbs free energy as a criterion for chemical reactions.

\section{System identification}
In this section, we introduce the method to identify the inner product and entropy via convex optimization, where we assume the Poisson bracket is known in advance.


\subsection{Problem formulation}
Consider a set of the measurement of the trajectory $\left\{ x_1, x_2, x_3, \cdots, x_m\right\}$ and the associated time derivative $\left\{\dot{x}_1, \dot{x}_2, \dot{x}_3, \cdots, \dot{x}_m \right\}$. Then our goal is to identify the metric $\mathcal{K}$ and the entropy $S$ to generate the systems that best matches the measurements:

\begin{problem}[Identifying the metricplectic system]
Suppose we have the measurements of system trajectory $\left\{ x_1, x_2, x_3, \cdots, x_m\right\}$ and its time derivative $\left\{\dot{x}_1, \dot{x}_2, \dot{x}_3, \cdots, \dot{x}_m \right\}$. We consider a cost function $r(S, \mathcal{K}  | \dot{x}_i, x_i ) \ge 0 $ and require that $r(S, \mathcal{K} | \dot{x}_i, x_i) = 0 $ when $\dot{x}_i = \left(\Pi(x_i) + K(x_i)\right)\left(\nabla H(x_i) + \nabla S(x_i)\right)$. Thus we have the optimization:

\begin{equation}\label{eq:cost}
     \min_{\mathcal{K}, S}\quad \sum^{m}_{i = 1} r(S, \mathcal{K} | \dot{x}_i, x_i)
\end{equation}
\end{problem}

We note that this optimization is defined for all functions $S$ and metric $\mathcal{K}$ that is challenging for numerical implementation in finte dimensional space.

\subsection{Systems identification via convex optimization}

\begin{algorithm}[t]
\caption{Bilevel convex optimization}
\label{alg:1}
\begin{algorithmic}
\footnotesize
\Require Measurement of system trajectory $\left\{x_1, x_2, x_3, \cdots, x_m\right\}$ and corresponding acceleration $\left\{\dot{x}_1, \dot{x}_2, \dot{x}_3, \cdots, \dot{x}_m\right\}$. Initial entropy $S^{(0)}$ parameterized by $\theta^{(0)}$. 
\For{iteration $t = 0, 1, \ldots$}
    \State {\texttt{// Metric search}}
    \State $\mathcal{K}^{(t)}, \theta^{(t)} \xleftarrow{\mathrm{LMI}}$ Problem \ref{prob:metric-search} with fixed $S = S^{(t)}$. 

    \State {\texttt{// Entropy search}}
    \State $\mathcal{S}^{(t)}, \psi^{(t)} \xleftarrow{\mathrm{LP}}$ Problem \ref{prob:entropy-search} with fixed $\mathcal{K} = \mathcal{K}^{(t)}$. 
\EndFor
\State {
\noindent \Return $\{ \mathcal{K}^{(t)}, S^{(t)} \}$
}
\end{algorithmic}
\end{algorithm}

\begin{algorithm}[t]
\caption{Stochastic bilevel convex optimiztion}
\label{alg:2}
\begin{algorithmic}
\footnotesize
\Require Measurement of system trajectory $\left\{x_1, x_2, x_3, \cdots, x_m\right\}$ and corresponding acceleration $\left\{\dot{x}_1, \dot{x}_2, \dot{x}_3, \cdots, \dot{x}_m\right\}$. Initial entropy $S^{(0)}$ parameterized by $\theta^{(0)}$.

\For{iteration $t = 0, 1, \ldots$}
    \State Randomly select mini batch for the trajectory $\left\{x_1, x_2, x_3, \cdots, x_q\right\}$ and acceleration $\left\{\dot{x}_1, \dot{x}_2, \dot{x}_3, \cdots, \dot{x}_q\right\}$ of size $q \le m$. 
    \State {\texttt{// Metric search}}
    \State $\mathcal{K}^{(t)}, \theta^{(t)} \xleftarrow{\mathrm{LMI}}$ Problem \ref{prob:metric-search} with fixed $S = S^{(t)}$. 

    \State {\texttt{// Entropy search}}
    \State $\mathcal{S}^{(t)}, \psi^{(t)} \xleftarrow{\mathrm{LP}}$ Problem \ref{prob:entropy-search} with fixed $\mathcal{K} = \mathcal{K}^{(t)}$. 
\EndFor
\State {
\noindent \Return $\{ \mathcal{K}^{(t)}, S^{(t)} \}$
}
\end{algorithmic}
\end{algorithm}

To approximate the $S$ and $\mathcal{K}$ in finite dimensional space, we formulate $S$ and $\mathcal{K}$ as polynomial function. Consider the polynomial basis of an $n$ dimensional system $x \in \mathbb{R}^n $ with order up to $r$:\
\begin{equation}
    \phi_r(x) = \left[1, x_1, x_2, \cdots, x_1^r, x_2^r, \cdots, x_n^r\right].
\end{equation}
We model the entropy as a polynomial function
\begin{equation}
    S(x) = \langle \psi, \phi_r(x)\rangle,
\end{equation}
and $K$ as a element-wise polynomial matrix:
\begin{equation}
    {K}(x)_{i, j} = \langle \theta_{i, j}, \phi_s(x) \rangle,
\end{equation}
with the vector of coefficients $\psi$ and $\theta$. 

We consider the cost function to the one-norm of the residual of the fitted metriplectic field, and we yield the following component-wise convex optimization for the system identification problem: 
\begin{problem}[Component-wise convex optimization]
\label{prob:bilevel-cvx}
    Consider that the metric $\mathcal{K}$ parameterized by $\theta = \left\{ \theta_{i, j} \right\}$ and the entropy $S$ parameterized by $\psi$. We minimize the one-norm of the discrepancy between the measured time derivative and the fitted metriplectic system:
    \begin{equation}
        \begin{aligned}
        \min_{ \psi, \theta }&\quad \sum_{i = 1}^m  \left|\dot{x}_i - \Pi(x_i)\left(\frac{\partial H + S}{\partial x}\right) - {K}(x_i)\left(\frac{\partial H + S}{\partial x}\right) \right|_1 \\
        \mathrm{s.t.}&\quad \ \ \eta I \succeq {K} \succeq 0, \\
                     &\quad -\eta \le \psi_k \le \eta, \quad \forall k,
        \end{aligned}
    \end{equation}
where $\eta$ is a real number that is sufficiently large to ensure the feasible set is compact and enable the globally optimal solutions to stay in the feasible set.  
\end{problem}

We note that Problem \ref{prob:bilevel-cvx} has convex feasible set, while the cost function is nonconvex as it involves only cross terms of the coefficients. However, we note that the problem is component-wise convex when either $\mathcal{K}$ or $S$ is fixed. To see this, when either of $\psi$ or $\theta$ is fixed, the cost function reduces to the one-norm of the linear function of the free variables and the feasible. Thus, we propose the following bilevel convex optimization algorithm that optimize the metric $\mathcal{K}$ and entropy $S$. 

For fixed entropy $S$, we note that the problem can be solved by the following convex optimization problem with Linear Matrix Inequality (LMI). Additional slack variables are also introduced to convert the non-smooth one-norm objective to an equivalent linear objective:
\begin{problem}[Metric search]
\label{prob:metric-search}
    Given fixed entropy $S = \bar{S} = \langle \bar{\psi}, \phi_r(x) \rangle$, we search the parameters for the metric via:
    \begin{equation}
    \begin{aligned}
        \min_{\theta, \delta} \quad &\sum_{i=1}^m \sum_{j=1}^n \delta_{i, j} \\ 
        \mathrm{s.t.}\quad & \eta I \succeq {K}(x) \succeq 0, \\ 
        & \delta_i \ge \dot{x}_i - \left(\Pi(x_i) + {K}(x_i)\right)\left(\frac{\partial H + \bar{S}}{\partial x}\right), \\
        - &\delta_i \le \dot{x}_i -\left(\Pi(x_i) + {K}(x_i)\right)\left(\frac{\partial H + \bar{S}}{\partial x}\right), \\
        &\delta_{i, j} \ge 0. \\
    \end{aligned}
    \end{equation}
    As the constraint $\eta I \succeq {K}(x) \succeq 0$ is imposed on all $x$, we introduce additional indeterminate variable $y$ to slack it as a sufficient sum of square (SOS) polynomial condition \cite{manchester2017control}:
    \begin{equation}
    \begin{aligned}
        y^{\transpose}{K}(x)y &\in SOS \Rightarrow y^{\transpose}{K}(x)y \ge 0 \Rightarrow {K}(x) \succeq 0, \\
        y^{\transpose}(\eta I - {K}(x))y &\in SOS \Rightarrow y^{\transpose}(\eta I - {K}(x))y \ge 0  \\
        &\Rightarrow \eta I \succeq {K}(x).
    \end{aligned}
    \end{equation}
\end{problem}
For fixed metric, we introduce the following Linear Programming (LP) to search the entropy:

\begin{problem}[Entropy search]
\label{prob:entropy-search}
Given fixed $\mathcal{K} = \bar{\mathcal{K}}$ defined by the component-wise polynomial matrix $\bar{{K}}(x)_{i, j} = \langle\bar{\theta}_{i, j}, \phi_s(x) \rangle $, we search the parameters for the entropy via:
\begin{equation}
    \begin{aligned}
        \min_{\psi, \delta} \quad &\sum_{i=1}^m \sum_{j = 1}^N \delta_{i, j} \\
        \mathrm{s.t.}\quad& \delta_i \ge \dot{x}_i - \left(\Pi(x_i) + \bar{{K}}(x_i)\right)\left(\frac{\partial H + S}{\partial x}\right), \\
        - &\delta_i \le \dot{x}_i - \left(\Pi(x_i) + \bar{{K}}(x_i)\right)\left(\frac{\partial H + S}{\partial x}\right), \\
        &\delta_{i, j} \ge 0, \\
        -&\eta \le \psi_k \le \eta, \quad \forall k \\
    \end{aligned}
\end{equation}
\end{problem}
Given these subproblems, we summarize the bilevel convex optimization problem in Algorithm \ref{alg:1}. We note that although the Problem \ref{prob:bilevel-cvx} is nonconvex though component-wise convex. Thus, the proposed bilevel optimization is not guaranteed to converge to the global optimum. To handle the nonconvexity, we introduce the stochastic version summarized in Algorithm \ref{alg:2}, that uses mini batch in each iterations to bring random search that makes it possible to escape the local minimum. The mini batch can also avoid large LP or LMI that bring challenges to the numerical solvers.

Though convergence to global optimum is not guaranteed, we show that total cost of the algorithm is guaranteed to converge, under the assumpition that each subproblem returns a feasible solution with cost lower than the previous iteration: 

\begin{theorem}[Convergence of Algorithm I]
\label{theorem:1}
    Assume that Problem \ref{prob:entropy-search} and \ref{prob:metric-search} returns a feasible solution with value no greater than previous iteration, the cost of Problem \ref{prob:bilevel-cvx} is guaranteed to converge. 
\end{theorem}
\begin{proof}
    Suppose cost function of Problem \ref{prob:bilevel-cvx} as $f(\theta, \psi)$, and $(\theta, \psi)$ is initialized at $(\theta^{(0)}, \psi^{(0)})$. By Algorithm \ref{alg:1}, we alternately minimize the cost with respect to $\theta$ and $\psi$. We use the superscript $^{(k)}$ to denote the variable after $k$-th update. Under the assumption that solving each subproblem results in a minimizer that is feasible and lower the cost from last iteration, we have the chain of inequality: $f(\theta^{(0)}, \psi^{(0)}) \ge f(\theta^{(1)}, \psi^{(0)}) \ge \cdots \ge  f(\theta^{(k)}, \psi^{(k-1)}) \ge f(\theta^{(k)}, \psi^{(k)}) \ge f(\theta^{(k+1)} \cdot, \psi^{(k)}) \cdots$.

    Denote that the smallest integer greater than or equal to $a$ as $\left\lceil a \right\rceil$. Then we have a monotone sequence of feasible solution $\{ (\theta^{( \left\lceil \frac{k}{2} \right\rceil )}, \psi^{(\left\lceil \frac{k-1}{2} \right\rceil )}) \}_{k \ge 0}$, with corresponding value of $f$ as $\{f_{k} \}_{k \ge 0}$. Thus we have that the sequence $\{f_k\}_{k \ge 0}$ monotonically decreases with trivial lower bound $0$, which  guarantees the existence of $f^*$ such that $\lim_{k\rightarrow\infty} f_k = f^*$ by Bolzano–Weierstrass theorem. 

\end{proof}




\section{Numerical Analysis}
In this section, we provide examples of the Hamiltonian system with the presented double bracket structure. We apply MOSEK to solve the LMI and LP in each subproblem in Algorithm \ref{alg:1}.

\subsection{Linear system}
We consider a 2-dimensional system with the momentum $p$, the configuration state $p$, and the Hamiltonian:
\begin{equation}
    H(q, p) = \frac{1}{2}p^2 + \frac{1}{2}q^2
\end{equation}
with the canonical Poission bracket:
\begin{equation}
    \Pi(p, q) = \begin{bmatrix}
        0 & 1 \\
        -1 & 0
    \end{bmatrix}.
\end{equation}

To illustrate the patterns of the free energy $E$, we consider the following polynomial function: 
\begin{equation}
\begin{aligned}
    E(p, q) &= -\frac{1}{2}(p^2 + 4p^4 - 4q^4)^2, \\
    S(p, q) &= E(p, q) - H.
\end{aligned}
\end{equation}
By construction, the zero level set of the free energy exhibit the $\infty$ shape that centered at the origin. 

We consider the symmetric bilinear map as defined by:
\begin{equation}
    {K} = \begin{bmatrix}
        1 & 0 \\
        0 & 1
    \end{bmatrix}. 
\end{equation}


\begin{figure}
    \centering
    \includegraphics[width=1\linewidth]{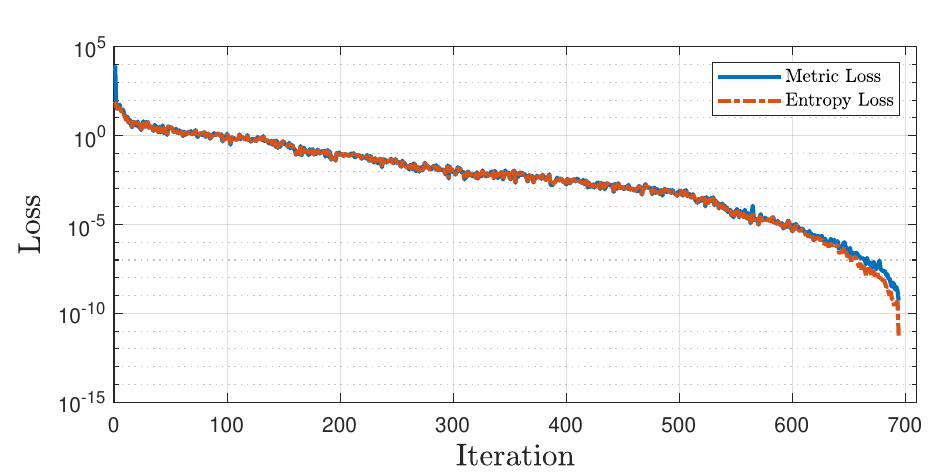}
    \caption{The convergence of the loss of Algorithm \ref{alg:2} for the 2D metriplectic system. The loss after solving each sub-problems are shown. Superlinear convergence rate is observed. As the loss converges to the trivial lower bound $0$, we conclude that the algorithm find the global optimizer in this case. }
    \label{fig:linear-cvx}
\end{figure}

The simulated trajectories in the phase space are shown in Fig. \ref{fig:linear-traj}. We show that the trajectory converges to the zero-level set of the free energy given the metriplectic structure.

We further apply the proposed system identification method to characterize the entropy and metric as the polynomial function given simulated trajectories. Due to the large number of measurements, we apply the stochastic bilevel convex optimization algorithm to avoid too many slack variables $\delta$. We consider to use quadratic polynomials to model the metric and $8$-th order polynomials to model the entropy. We show that the proposed method converges to the ground truth value. The evolution of the loss function is shown in Fig. \ref{fig:linear-cvx}. The total number of slack variables we introduce is 4000 with mini batch size 2000. 

\subsection{$\mathrm{SO}(3)$ system}

\begin{figure}
    \centering
    \includegraphics[width=0.65\linewidth]{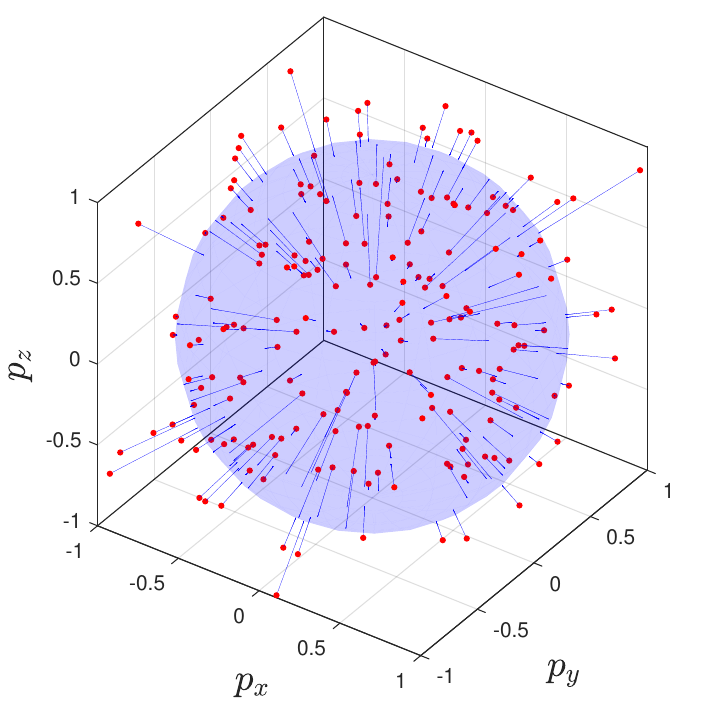}
    \caption{Trajectory of the metriplectic system on $\mathrm{SO}(3)$. The trajectories converges to the zero-level set of the free energy as a sphere. }
    \label{fig:SO3-traj}
\end{figure}

\begin{figure}
    \centering
    \includegraphics[width=1\linewidth]{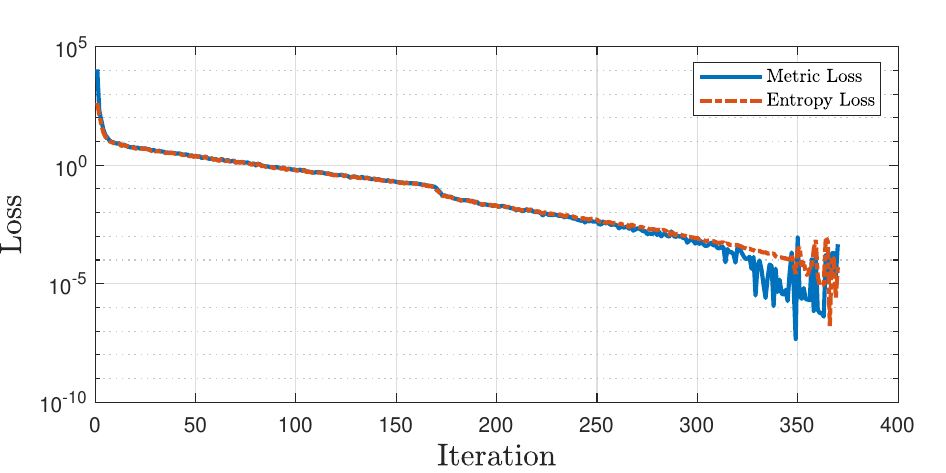}
    \caption{The convergence of the loss of Algorithm \ref{alg:2} for the metriplectic system on $\mathrm{SO}(3)$. The loss after solving each sub-problems are shown. The chattering at the last few iterations are due to the numerical issue of SDP solvers. }
    \label{fig:SO3-cvx}
\end{figure}

Now we apply the framework for a metriplectic system on $\mathrm{SO}(3)$. Consider the Lagrangian in the Lie algebra of $\mathrm{SO}(3)$, i.e, $L : \mathfrak{so}(3) \rightarrow \mathbb{R} $, and its corresponding Euler-Poincar\'e equations:
\begin{equation}
    \frac{d}{dt}\left( \frac{d L}{d \xi} \right) = \mathrm{ad}_{\xi}^* \frac{d L}{d \xi},
\end{equation}
where $\xi \in \mathfrak{so}(x)$. For a rigid body in 3D space, we have the Lagrangian
\begin{equation}
    L = \frac{1}{2}\xi^{\transpose}I\xi = \frac{1}{2}p^{\transpose}I^{-1}p,
\end{equation}
where the momentum is $p = I\xi$ and $I$ is the inertia matrix. 
Then we have the equation of motion as:
\begin{equation}
    \dot{p} = \Pi(p)\nabla H,
\end{equation}
with the cross product Poisson bracket and the Hamiltonian defined respectively by:
\begin{equation}
    \Pi(p) = p^{\times}, H = L, \nabla H = \xi.
\end{equation}
where $(\cdot)^{\times}$ satisfy $a^{\times}b = a\times b, \forall a, b \in \mathbb{R}^3$.
We consider the free energy to be:
\begin{equation}
    E(p) = - \frac{1}{2}(p_x^2 + p_y^2 + p_z^2 - 1)^2
\end{equation}
and the corresponding entropy becomes:
\begin{equation}
    S(p) = E(p) - H(q).
\end{equation}
The metric is chosen as the identify matrix:
\begin{equation}
    {K} = \begin{bmatrix}
        1 & 0 & 0 \\
        0 & 1 & 0 \\
        0 & 0 & 1 \\
    \end{bmatrix}. 
\end{equation}
The simulated trajectories of the system and the result of the system identification are presented in Fig. \ref{fig:SO3-traj} and Fig. \ref{fig:SO3-cvx}. We show that the sampled trajectories all converged to the zero level set of the free energy. We quadratic polynomial to model the metric and $6$-th order polynomial to model the entropy. The orders of the polynomials are both higher than the ground truth, which are sufficient to represent the true systems but also makes it more challenging for numerical optimization. The size of the mini batch is 2000, resulting in 6000 slack variables. The cost suggest that the average of each slack variable is negligible after 350 iterations in Algorithm 2. 

\section{Discussion and future work}
In this work, we have shown that the proposed bilevel convex optimization can successfully identify the underlying dynamics via solving an LMI and LP alternately. Though it is guaranteed that the total cost will converge assuming that each convex subproblem can solved by the convex optimization solver with feasible solutions to lower the entire cost, it is not guaranteed that the solution will uniquely converge to a minimizer. The non-uniqueness is due to the possible non-unique metriplectic representations. For a polynomial representation of the metric and entropy with a fixed order, multiple solutions may yield identical polynomial vector fields. In future work, it is critical to analyze the regularity conditions for the uniqueness and identifiability of the metriplectic systems. As we leverage the off-the-shelf solver for each subproblem, we are not able to claim any convergence for the sequence of the feasible solutions shown in Theorem~\ref{theorem:1} without details of the numerical solvers. Future work will investigate the coordinate descent \cite{wright2015coordinate} to leverage the component-wise convexity. 

As the robotics systems can be naturally modeled by rigid body motions that is a Hamiltonian system evolving on matrix Lie groups \cite{ghaffari2022progress}, considering the dissipative effects induced by the generalized formalism in control~\cite{jang2023convex, Teng-RSS-23, teng2022lie, teng2022error} and realtime state estimation~\cite{yu2023fully, teng2021legged, teng2022input, he2024legged} on Lie groups could be interesting future work.

\section{Conclusion}
\label{sec:conclusion}
This paper generalizes the classical metriplectic formalism to Hamiltonian systems with nonconservative dissipation. Via relaxation the conditions for entropy and metric, we induces the \emph{free energy} function as the summation of entropy and the total mechanical energy of the system. We show that under the double bracket condition, the free energy changes monotonically. We further provided a bilevel convex optimization approach for system identification of the metriplectic system. The proposed system identification approach is demonstrated in numerical experiments to recover the underlying entropy and brackets given the observation of the systems. 

\section*{Acknowledgments}
This work was supported by AFOSR MURI FA9550-23-1-0400. The authors thank Dr. Frederick A. Leve for his encouragement and support. Additional support from AFOSR FA9440-31-1-0215 and NSF DMS 2103026.

{\footnotesize 
\balance
\bibliographystyle{IEEEtran}
\bibliography{bib/strings-abrv,bib/ieee-abrv,bib/references}
}

\end{document}